\documentclass{article}
\usepackage{a4}
\usepackage[T1]{fontenc}
\usepackage{ae,aecompl}

\usepackage{amssymb}
\usepackage{amsmath}

\usepackage{algorithm}
\usepackage{algorithmic}

\usepackage{vaucanson-g}

\newcommand{\PH}[1]{PH(#1)}

\newtheorem{definition}{Definition}
\newtheorem{theorem}{Theorem}
\newtheorem{lemma}{Lemma}
\newenvironment{proof}{\noindent{\bf Proof:}}{\hfill$\Box$}

\title{On-line construction of position heaps\footnote{A preliminary version of this paper has been presented to the 18th International Symposium on String Processing and Information Retrieval (SPIRE), Pisa (Italy), August 2011}}
\author{Gregory Kucherov\thanks{Universit\'e Paris-Est \& CNRS,
    Laboratoire d'Informatique Gaspard Monge, Marne-la-Vall\'ee,
    France, {\tt Gregory.Kucherov@univ-mlv.fr}} \thanks{Department of
    Computer Science, Ben-Gurion University of the Negev, Be'er Sheva, Israel}}
\date{\empty}

\begin{document}

\maketitle

\begin{abstract}
We propose a simple linear-time on-line algorithm for constructing a
position heap for a string \cite{Ehrenfeucht2011100}. Our definition
of position heap differs slightly from the one proposed in
\cite{Ehrenfeucht2011100} in that it considers the suffixes ordered 
in the descending order of length. Our construction is based on classic suffix pointers
and resembles Ukkonen's algorithm for suffix trees
\cite{Ukkonen93}. Using suffix pointers, the position
heap can be extended into the augmented position heap that allows for
a linear-time string matching algorithm \cite{Ehrenfeucht2011100}. 
\end{abstract}

\section{Introduction}
The theory of string algorithms developed beautiful data structures
for string matching and text indexing. Among them, {\em suffix tree}
and {\em suffix array} are most widely used structures, providing 
efficient solutions for a wide range of applications
\cite{CrochemoreRytter94,books/cu/Gusfield1997}. The DAWG ({\em Directed Acyclic Word
  Graph}) \cite{BlumerAtAlTCS85}, also known as {\em suffix automaton} \cite{CrochemoreTCS86}, is another elegant
structure that can be used both as a text index \cite{BlumerAtAlTCS85} or as a
matching automaton \cite{CrochemoreMFCS88,CrochemoreRytter94}. 

Recently, a new {\em position heap} data structure was proposed
\cite{Ehrenfeucht2011100}. 
Similar to the suffix tree, DAWG or suffix array, position heap allows for a
pre-processing of a text string in order to efficiently search for patterns in
it. As for the above-mentioned data structures, a position heap for a
string of length $n$ can be constructed in time $O(n)$. Then all
locations of a 
pattern of length $m$ can be found in 
time $O(m+occ)$, where $occ$ is the number of occurrences.  

The construction algorithm of \cite{Ehrenfeucht2011100} processes the
string from right to left, like Weiner's algorithm does for suffix
trees \cite{Weiner73}. Moreover, the construction requires a
so-called dual heap, which is an additional trie on the same
set of nodes. The position heap and its dual heap are constructed
simultaneously. 

To obtain a linear-time pattern matching algorithm of
\cite{Ehrenfeucht2011100}, the position heap should be post-processed
in order to add some additional information, resulting in the {\em
  augmented position heap}. The most important element of this information
includes so-called {\em maximal-reach pointers} assigned to certain
nodes. Computing these pointers makes use of the dual heap too. 

In this paper, we propose a different construction of the position
heap. First, we change the definition of the position heap by
reversing the order of suffixes and thus allowing for the left-to-right
traversal of the input string. The modified definition, however,
preserves good properties of the position heap and does not affect the string
matching algorithm proposed in \cite{Ehrenfeucht2011100}. For this
modified definition, we propose an {\em on-line} algorithm for
constructing the position heap. Our algorithm does not use the dual
heap, replacing it by classic {\em suffix pointers} used for
constructing suffix trees by Ukkonen's algorithm \cite{Ukkonen93}
or for constructing the DAWG \cite{BlumerAtAlTCS85}. Our algorithm is
simple and can be compared to Ukkonen's algorithm for suffix
trees, as opposed to Weiner's algorithm that constructs the suffix tree
by inserting suffixes right-to-left (i.e. shortest first). We deliberately use some
terminology of Ukkonen's algorithm to underline this similarity. 

We further show that the augmented position heap can be easily
constructed using suffix pointers. Thus, we completely eliminate the
use of the dual heap, replacing it by suffix pointers for constructing
both the position heap and its augmented version. Even if this
replacement does not provide an immediate improvement in space or
running time, we believe that our construction is conceptually simpler
and more natural.

\medskip
Throughout the paper, we assume we are given a constant-size alphabet $A$. Positions of
strings over $A$ are numbered from 1, that is,  a string $w$ of
\emph{length} $k$ is $w[1] \ldots w[k]$. 
The length~$k$ of $w$ is denoted by $|w|$. 
$w[i..j]$ denotes substring $w[i] \ldots w[j]$.

A {\em trie} (term attributed to Fredkin \cite{Fredkin60}) is a
simple natural data structure for storing a set of strings. It is a
tree with edges labeled by alphabet letters, such that for any
internal node, the edges leading to the children nodes are labeled by
distinct letters. In this paper, we assume the edges to be directed
towards leaves, and call an edge labeled by a letter $a$ an $a$-edge. 
A {\em label} of a node ({\em path label}) is the string formed by the
letters labeling the edges of the path from the root to this
node. Given a trie, a string $w$ is said to be represented in the trie
if it is a path label of some node. The corresponding node will then
be denoted by $\overline{w}$. 
% Mixing nodes and labels. Index a string into a trie. 

\section{Definition of position heap}

To define position heaps, we first need to introduce
the {\em sequence hash tree} proposed by Coffman and Eve back in 1970
\cite{CoffmanEve70} as a data structure for implementing hash
tables. Assume we are given an ordered set of strings
$W=\{w_1,\ldots,w_n\}$ and assume for now that no $w_i$ is a prefix of
$w_j$ for any $j<i$. The sequence hash tree for $W$, denoted $SHT(W)$, is a trie
defined by the following iterative construction. We start with
the tree $SHT_0(W)$ consisting 
of a single root node $\mathit{root}$\footnote{This
  definition agrees with the definition of \cite{CoffmanEve70} but is slightly
  different from that of
  \cite{Ehrenfeucht2011100} which defines the root to store $w_1$. The
  difference is insignificant, however.}. We then
construct $SHT(W)$ by processing strings 
$w_1,\ldots,w_k$ in this order and for each
$w_i$, adding one node to the tree. By induction, assume that $SHT_i(W)$
is the sequence hash tree for 
$\{w_1,\ldots,w_i\}$. To construct $SHT_{i+1}(W)$, we find the shortest
prefix $v$ of $w_{i+1}$ which is {\em not} represented in
$SHT_i(W)$. Note that by our assumption, such a prefix always
exists. Let $v=v'a$, $a\in A$, i.e. $v'$ is the longest prefix of 
$w_{i+1}$ represented in $SHT_i(W)$. Then $SHT_{i+1}(W)$ is obtained from
$SHT_i(W)$ by adding a new node as a child of $v'$ connected to $v'$ by an 
$a$-edge and pointing to $w_{i+1}$. After inserting all strings of
$W$, we obtain $SHT(W)$, that is $SHT(W)=SHT_k(W)$. Thus, $SHT(W)$ is a trie of $n+1$ nodes such that a node pointing to $w_i$ is
labeled by some prefix of $w_i$. Note that the size of the sequence
hash tree depends only on the number of strings in the set and does
not depend on the length of those. An example of sequence hash
tree is given on Figure~\ref{SHT-figure}. 

\begin{figure}
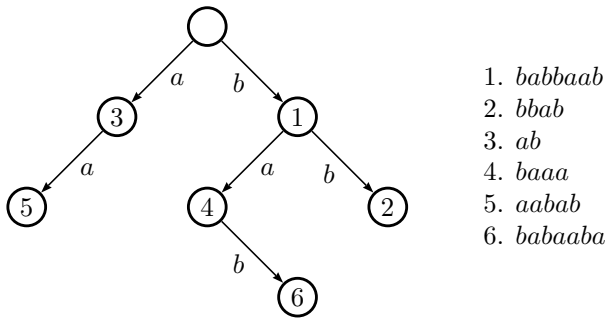

\begin{center}
\VCDraw{%
\begin{VCPicture}{(0,0)(10,6)}
% states
\MediumState
\State{(4,6)}{root}
\State[1]{(6,4)}{1}
\State[2]{(8,2)}{2}
\State[3]{(2,4)}{3}
\State[4]{(4,2)}{4}
\State[5]{(0,2)}{5}
\State[6]{(6,0)}{6}
\EdgeL{root}{3}{a}
\EdgeL{3}{5}{a}
\EdgeR{root}{1}{b}
\EdgeL{1}{4}{a}
\EdgeR{1}{2}{b}
\EdgeR{4}{6}{b}
\end{VCPicture}
}
\begin{minipage}{3cm}
1. $babbaab$\\
2. $bbab$\\
3. $ab$\\
4. $baaa$\\
5. $aabab$\\
6. $babaaba$
\end{minipage}
\end{center}
\caption{Sequence hash tree for the set of strings shown on the
  right. Each node stores the rank of the corresponding string in the
set.}
\label{SHT-figure}
\end{figure}

We now define the {\em position heap} of a
string $T$. In \cite{Ehrenfeucht2011100}, the position heap for $T$ is
defined as the sequence hash tree for the set of suffixes of $T$,
where {\em the suffixes are ordered in the ascending order of
  length}, i.e. from right to left. This insures, in particular, the condition that no suffix is a prefix of
a previously inserted suffix, and then no suffix is already represented in the
position heap at the time of its insertion. 
% Therefore, the construction is well-defined.

%% As every suffix of $T$ is naturally identified by its start position,
%% we will assume that each node of a position heap stores a position of
%% $T$. By a slight abuse of language, we will sometimes say that a node
%% ``stores a suffix'' 
%% % or ``points to a suffix'' 
%% referring actually to
%% the starting position of that suffix. 

In this paper, we define the position heap of $T$ to be the sequence
hash tree for the set of suffixes of $T$, 
where {\em the suffixes are ordered in the descending order of
  length}, i.e. from left to right. From now on, we stick to this order. 
An immediate observation is that the assumption of the suffix hash
tree does not hold anymore, and it may occur that an inserted suffix
is already represented in the position heap by an existing node. One easy way to cope with
this is to systematically assume that $T$ is ended by a special
sentinel symbol $\$$, like it is generally assumed for the suffix tree. 

On the other hand, as we will be interested in an on-line construction
of the position heap, we will still need to construct the position heap for
strings without the ending sentinel symbol. For that, we have to
slightly change the definition of sequence hash tree of a set $W$, by
allowing one node to point to several strings of $W$. The definition
of the position heap extends then to any string, with the only
difference that inserting a suffix may no longer lead to the creation
of a new node, but to adding a pointer to this suffix to an existing node. 
This feature,
however, will be used in a very restricted way, as the following
observation shows. 

\begin{lemma}
\label{one-or-two}
Let $W$ be a set of distinct strings. Then every node of $SHT(W)$
points to at most two strings of $W$. 
\end{lemma}
\begin{proof}
The only situation when a new pointer gets inserted to an existing
node is when the inserted string $w_{i+1}$ is already represented in
$SHT_i(W)$. Since all strings of $W$ are distinct, this situation may
occur only once for each node. Therefore, each node of $SHT(W)$ points
to one or two strings of $W$. 
% Straightforward from the definition of $SHT(W)$ and the fact that all
% strings are distinct. 
\end{proof}

As a consequence of Lemma~\ref{one-or-two}, a position heap contains
two types of 
nodes, pointing respectively to one and two suffixes of $T$. The
former will be called {\em regular nodes} and the latter {\em double
  nodes}. 
We naturally assume that a pointer to a suffix is simply the
starting position of that suffix, therefore regular and double nodes
store one and two string positions respectively. Hereafter we
interchangeably refer to ``suffixes'' and ``positions'' when the
underlying string is unambiguously defined. 

Figure~\ref{PH-example} provides an example of a position heap. 
\begin{figure}
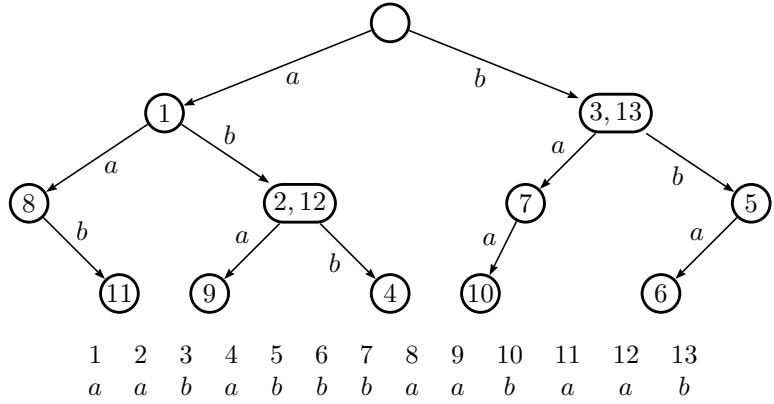

\begin{center}
\VCDraw{%
\begin{VCPicture}{(0,-1)(16,8)}
% states
\MediumState
\State{(8,6)}{root} 
\State[1]{(3,4)}{1} 
\StateVar[2,\textrm{12}]{(6,2)}{2}
\StateVar[3,\textrm{13}]{(13,4)}{3}
\State[4]{(8,0)}{4}
\State[5]{(16,2)}{5}
\State[6]{(14,0)}{6}
\State[7]{(11,2)}{7}
\State[8]{(0,2)}{8}
\State[9]{(4,0)}{9}
\State[10]{(10,0)}{10}
\State[11]{(2,0)}{11}
%\State[12]{(6,0)}{12}
%\State[13]{(6,0)}{13}
% transitions
\EdgeL{root}{1}{a} 
\EdgeL{1}{8}{a} 
\EdgeL{8}{11}{b} 
\EdgeL{1}{2}{b} 
\EdgeR{2}{9}{a} 
\EdgeR{2}{4}{b} 
\EdgeR{root}{3}{b} 
\EdgeR{3}{7}{a} 
\EdgeR{7}{10}{a} 
\EdgeR{3}{5}{b} 
\EdgeR{5}{6}{a} 
\end{VCPicture}
}

\begin{tabular}{ccccccccccccc}
1&2&3&4&5&6&7&8&9&10&11&12&13\\
$a$&$a$&$b$&$a$&$b$&$b$&$b$&$a$&$a$&$b$&$a$&$a$&$b$
\end{tabular}
\end{center}
\caption{Position heap for string $aababbbaabaab$. Double nodes store
  pairs of positions.}
\label{PH-example}
\end{figure}

\section{Properties of position heap}
\label{sect-properties}

Denote by $\PH{T}$ the position heap for a string $T[1..n]$
as defined in the previous section. In the following
theorem, we summarize some key properties of the position
heap. 

\begin{theorem}[\cite{Ehrenfeucht2011100}]
\label{PH-properties}
Consider $\PH{T[1..n]}$. The following properties hold.
\begin{itemize}
\item[\textit{(i)}] A substring $w$ of $T$ is represented in $\PH{T}$ iff
  $T$ contains occurrences of strings $w[1..1]$, $w[1..2]$,
  $w[1..3]$, \ldots, $w[1..|w|]$, appearing at increasing
  positions in this order. 
\item[\textit{(ii)}] The labels of all nodes of $\PH{T}$ form a factorial set. That
  is, if a string is represented in $\PH{T}$, all its substrings are
  represented too.
\item[\textit{(iii)}] The depth of $\PH{T}$ is no more than $2h(T)$, where $h(T)$ is the length of the longest
substring $w$ of $T$ which occurs $|w|$ times in $T$ (possibly with 
overlap). 
\item[\textit{(iv)}] If a substring $w$ occurs in $T$ at least $|w|$ times, then $w$ is
  represented in $\PH{T}$. Inversely, if $w$ is not represented in
  $\PH{T}$ and $w'$ is the longest prefix of $w$ which is represented,
  then $w$ cannot occur in $T$ more than $|w'|$ times. 
\end{itemize}
\end{theorem}
\begin{proof}
% is straightforward from the definition. Indeed,
\textit{(i)} The 'if'-part follows immediately from the definition of
$\PH{T}$ and the left-to-right order of suffixes. 
If a substring $w$ is represented in $\PH{T}$, then nodes
$\overline{w[1..1]}$, $\overline{w[1..2]}$, $\overline{w[1..3]}$,
\ldots, $\overline{w[1..|w|]}$ have been created in this respective
order. The creation of each such node $\overline{w[1..\ell]}$ has
been triggered by an insertion of a suffix starting with
$w[1..\ell]$. Since suffixes are inserted from left to right,
property \textit{(i)} follows. 

Properties \textit{(ii)}-\textit{(iv)} have been established in \cite{Ehrenfeucht2011100} but
remain valid for our definition of position heap when suffixes are
inserted from left to right. Actually, these properties are valid for
{\em any} order of inserting suffixes into the position heap. 

\textit{(ii)} It is sufficient to show that if some string $w[1..\ell]$
is represented in $\PH{T}$, then both $w[i..\ell-1]$ and $w[2..\ell]$ are
represented too. For $w[1..\ell-1]$, this is obvious from
construction. For $w[2..\ell]$, this can be seen from
Property~\textit{(i)}. Indeed, if strings $w[1..1]$, $w[1..2]$,
  $w[1..3]$, \ldots, $w[1..\ell]$ appear in $T$ in this relative order,
then we have strings $w[2..2]$, $w[2..3]$, \ldots, $w[2..\ell]$
appearing in $T$ at increasing positions too. By Property~\textit{(i)},
this ensures that $w[2..\ell]$ is represented in $\PH{T}$. 

\textit{(iii)} Let $\overline{w}$ be one of the deepest
nodes of $\PH{T}$, i.e. the depth of $\PH{T}$ is $d=|w|$. From
Property~\textit{(i)}, strings $w[1..\lceil d/2\rceil]$, $w[1..\lceil
  d/2\rceil+1]$, \ldots, $w[1..d]$ occur at $T$ at distinct positions,
and therefore 
$w[1..\lceil d/2\rceil]$ % of length $\lceil d/2\rceil$
occurs at least $\lceil d/2\rceil$ times in $T$. Then $h(T)\geq \lceil
d/2\rceil$ and the depth $d$ of $\PH{T}$ is bounded by $2h(T)$. 

\textit{(iv)} If a substring $w$ occurs in $T$ at least $|w|$ times,
then there exist successive occurrences of $w[1..1]$, $w[1..2]$,
$w[1..3]$,\ldots,$w[1..|w|]$, and, by Property~\textit{(i)}, $w$ is
  represented in $\PH{T}$. Assume now that $w$ is not represented in
  $\PH{T}$ and $w'$ is the longest prefix of $w$ which is
  represented. Assume further that $a$ is the letter that follows
  prefix $w'$ in $w$. Observe that $w'a$ occurs at most $|w'|$ times,
  as the contrary would mean that $w'a$ is represented too, which
  contradicts the choice of $w'$. Therefore, $w'a$ occurs at most
  $|w'|$ times and so does $w$. 
\end{proof}

Properties \textit{(iii)} and \textit{(iv)} show that the position heap of a
string ``adapts'' to the frequencies of its substrings. In particular, if
a string is ``frequent'' (occurs as many times as it is long), then it
is necessarily represented in the position heap. On the other hand, if
it is not represented, it has less occurrences than its length. The
latter property is crucial for obtaining a linear-time string matching
algorithm of \cite{Ehrenfeucht2011100}. 

\section{On-line construction algorithm}
\label{sect-algo}

Let us have a closer look
at the properties of double nodes of a position heap $\PH{T}$. Each
such node stores two positions
$i,j$ of $T$. Assume $i<j$, then positions $i$ and $j$ will be called the {\em
  primary} and the {\em secondary} positions respectively. 

\begin{lemma}
\label{next-secondary}
Let $T=T[1..n]$. If $j<n$ is the secondary position of some node of 
% a position heap, 
$\PH{T}$,  
then so is $j+1$.
\end{lemma}
\begin{proof}
Consider $\PH{T}$ for some string $T[1..n]$. Assume $i,j$, $i<j$, are
respectively primary and secondary positions of some node. This means
that by the time the suffix $T[j..n]$ is inserted into $\PH{T}$ during its
construction, node $\overline{T[j..n]}$ already exists. By 
Theorem~\ref{PH-properties}\textit{(ii)}, node $\overline{T[j+1..n]}$
exists too. {\em A fortiori}, 
%% This means, in
%% turn, that there exist suffixes starting with prefixes $T[j]$,
%% $T[j]T[j+1]$, $T[j]T[j+1]T[j+2]$, \ldots, $T[j..n-1]$, $T[j..n]$, and
%% all these suffixes start at positions smaller than $j$. Then there exist suffixes starting with prefixes 
%% $T[j+1]$, $T[j+1]T[j+2]$, \ldots, $T[j+1..n-1]$, $T[j+1..n]$, and
%% all these suffixes start at positions smaller than $j+1$. 
node $\overline{T[j+1..n]}$ exists when $T[j+1..n]$ is inserted into
$\PH{T}$. Therefore, $j+1$ becomes the
secondary position of that node after the insertion of suffix $T[j+1..n]$. 
\end{proof}

Lemma~\ref{next-secondary} implies that all positions of $T[1..n]$ are
split into two intervals: primary positions $[1..s-1]$, for some
position $s$, and secondary positions $[s..n]$. Position $s$ will be
called {\em active secondary position}, or {\em active position} for
short. 

\medskip
Assume we have constructed the position heap $\PH{T[1..k]}$ for some
prefix $T[1..k]$ of the input string $T[1..n]$. Let us analyze the
differences between $\PH{T[1..k]}$ and $\PH{T[1..k+1]}$ and the 
modifications that need to be made to transform the former into the
latter. 

Let 
% $a=T[k+1]$ and let 
$s$ be the active position of $T[1..k]$. First
observe that for suffixes $1,\ldots,s-1$, no changes need to be
made. Inserting each suffix $T[i..k]$ for $1\leq i\leq s-1$ into
$\PH{T[1..k]}$ led to the creation of a new node. This means that by
the time this suffix was inserted into $\PH{T[1..k]}$, some
prefix $T[i..\ell]$ of $T[i..k]$, $\ell\leq k$, was not represented in the position heap,
which led to the creation of a new node $\overline{T[i..\ell]}$
with the minimal such $\ell$. 
% position $\ell<k$ such that the node $\overline{T[1..\ell]}$ does not
%% exist, while for some $\ell<k$ but not the node
%% $\overline{T[1..\ell+1]}$ which is created as a result of the
%% insertion. The same will occur when inserting suffix $T[i..k+1]$ into
%% $\PH{T[1..k+1]}$. 
%% of $1,\ldots,k-1$ creates a new node in
%% $\PH{T[1..k]}$, therefore considered as suffixes of 
This shows that inserting suffixes $1,\ldots,s-1$ involve completely identical
steps in the construction of both $\PH{T[1..k]}$ and
$\PH{T[1..k+1]}$. 

\begin{figure}
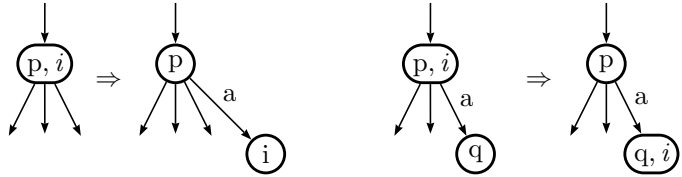

\begin{center}
\VCDraw{%
\begin{VCPicture}{(-1,-2)(1,1)}
\StateVar[$p$,\textit{i}]{(0,0)}{secondary}
\Initial[n]{secondary}
\HideState
\State{(-1,-2)}{left}\State{(0,-2)}{middle} \State{(1,-2)}{right}
\ShowState
\EdgeR{secondary}{left}{} 
\EdgeR{secondary}{middle}{} 
\EdgeL{secondary}{right}{} 
\end{VCPicture}
}
$\Rightarrow$
\VCDraw{%
\begin{VCPicture}{(-1,-2)(1,1)}
\State[$p$]{(0,0)}{secondary}
\Initial[n]{secondary}
\HideState
\State{(-1,-2)}{left}\State{(0,-2)}{middle} \State{(1,-2)}{right}
\ShowState
\State[$i$]{(2,-2)}{new}
\EdgeR{secondary}{left}{} 
\EdgeR{secondary}{middle}{} 
\EdgeL{secondary}{right}{} 
\EdgeL{secondary}{new}{$a$}
\end{VCPicture}
}
\hspace*{2cm}
\VCDraw{%
\begin{VCPicture}{(-1,-2)(2,1)}
\StateVar[$p$,\textit{i}]{(0,0)}{secondary}
\Initial[n]{secondary}
\HideState
\State{(-1,-2)}{left}\State{(0,-2)}{middle} 
\ShowState
\State[$q$]{(1,-2)}{right}
\EdgeR{secondary}{left}{} 
\EdgeR{secondary}{middle}{} 
\EdgeL{secondary}{right}{$a$} 
\end{VCPicture}
}
$\Rightarrow$
\VCDraw{%
\begin{VCPicture}{(-1,-2)(2,1)}
\State[$p$]{(0,0)}{secondary}
\Initial[n]{secondary}
\HideState
\State{(-1,-2)}{left}\State{(0,-2)}{middle} 
\ShowState
\StateVar[$q$,\textit{i}]{(1,-2)}{right}
\EdgeR{secondary}{left}{} 
\EdgeR{secondary}{middle}{} 
\EdgeL{secondary}{right}{$a$} 
\end{VCPicture}
}
\end{center}
\caption{Updating secondary position $i$ when transforming $\PH{T[1..k]}$ into $\PH{T[1..k+1]}$: {\em first case} (left) and {\em second case} (right)}
\label{cases12}
\end{figure}

The situation is different for the secondary positions
$s,\ldots,k$. Each suffix $T[i..k]$ for $s\leq i\leq k$ was already
represented in $\PH{T[1..k]}$ at the moment of its insertion, and then
resulted in the addition of the secondary position $i$ to the
node $\overline{T[i..k]}$. When inserting the corresponding suffix
$T[i..k+1]$ into the position heap 
$\PH{T[1..k+1]}$, two cases arise. 
% \textit{(i)} 
In the {\em first case}, 
inserting the suffix $T[i..k+1]$ leads to the
creation of the new node $\overline{T[i..k+1]}$ if this node does not
exist yet. Position $i$ then becomes the primary position of this
new node. 
Observe that this only occurs when $\PH{T[1..k]}$ does not
contain an $T[k+1]$-edge outgoing from the node
$\overline{T[i..k]}$. It is easily seen that such an edge cannot
appear by the time of insertion of $T[i..k+1]$ into $\PH{T[1..k+1]}$
if it was not already present in $\PH{T[1..k]}$. In the {\em second
  case}, node $T[i..k]$ has an outgoing $T[k+1]$-edge in
$\PH{T[1..k]}$, and in the construction of $\PH{T[1..k+1]}$, the
secondary position $i$ stored in this node should be ``moved'' to the
child node $\overline{T[i..k+1]}$. It becomes then the secondary
position of this node. The two cases are illustrated in Figure~\ref{cases12}. 

Observe now that if for a secondary position $i$, the corresponding node
$\overline{T[i..k]}$ has 
an outgoing $T[k+1]$-edge, then so does the node
$\overline{T[i+1..k]}$ storing the secondary position $i+1$. This can
again be seen from the factorial property of the position heap
(Theorem~\ref{PH-properties}\textit{(ii)}). This shows that the above
two cases split the 
interval of secondary positions $[s..k]$ into two
subintervals $[s..t-1]$ and $[t..k]$, such that node
$\overline{T[i..k]}$ does not have an outgoing $T[k+1]$-edge for 
$i\in [s..t-1]$ and does have such an edge for $i\in [t..k]$. 

The above discussion is summarized in the following lemma specifying
the changes that have to be made to 
transform $\PH{T[1..k]}$ into $\PH{T[1..k+1]}$. 

\begin{lemma}
\label{main}
Given $T[1..n]$, consider $\PH{T[1..k]}$ for $k<n$. Let $s$ be the
active secondary position, stored in the node
$\overline{T[s..k]}$. Let $t\geq s$ be the smallest position such that
node $\overline{T[t..k]}$ has an outgoing $T[k+1]$-transition. To
obtain $\PH{T[1..k+1]}$, $\PH{T[1..k]}$ should be modified in the
following way:
\begin{itemize}
\item[\textit{(i)}] for every node $\overline{T[i..k]}$, $s\leq i\leq
  t-1$, create a new child linked to $\overline{T[i..k]}$ by a
  $T[k+1]$-edge. Delete secondary position $i$ from the node
  $\overline{T[i..k]}$ and assign it as a primary position to the new
  node $\overline{T[i..k+1]}$,
\item[\textit{(ii)}] for every node $\overline{T[i..k]}$, $i\geq t$,
  move the secondary position $i$ from node $\overline{T[i..k]}$ to
  node $\overline{T[i..k+1]}$. 
\end{itemize}
\end{lemma}

\medskip
We describe now the algorithm implementing the changes specified by
Lemma~\ref{main}. 
We augment $\PH{T}$ with {\em suffix pointers} $f$ defined in the usual way:
\begin{definition}
\label{suf-link}
For each node $\overline{T[i..j]}$ of $\PH{T}$, a {\em suffix pointer} is
defined by \linebreak[4]$f(\overline{T[i..j]})=\overline{T[i+1..j]}$. 
\end{definition}
Note that the definition is sound, as the node $\overline{T[i+1..j]}$
exists whenever the node $\overline{T[i..j]}$ exists, according to
Theorem~\ref{PH-properties}\textit{(ii)}. For the root 
node, it will be convenient for us to define $f(\mathit{root})=\perp$, where
$\perp$ is a special node such that there is an $a$-edge between
$\perp$ and $\mathit{root}$ for every $a\in A$ (similar to Ukkonen's algorithm
\cite{Ukkonen93}). Figure~\ref{PH-with-suf-point} shows the
position heap of Figure~\ref{PH-example} supplemented by suffix pointers. 
\begin{figure}
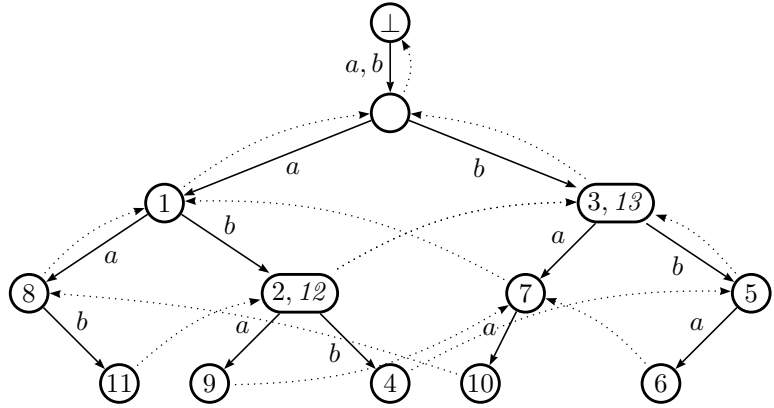

\begin{center}
\VCDraw{%
\begin{VCPicture}{(0,-1)(16,8)}
% states
\MediumState
\State[\perp]{(8,8)}{bottom}
\State{(8,6)}{root} 
\State[1]{(3,4)}{1} 
\StateVar[2,\mathit{12}]{(6,2)}{2}
\StateVar[3,\mathit{13}]{(13,4)}{3}
\State[4]{(8,0)}{4}
\State[5]{(16,2)}{5}
\State[6]{(14,0)}{6}
\State[7]{(11,2)}{7}
\State[8]{(0,2)}{8}
\State[9]{(4,0)}{9}
\State[10]{(10,0)}{10}
\State[11]{(2,0)}{11}
%\State[12]{(6,0)}{12}
%\State[13]{(6,0)}{13}
% transitions
\EdgeR{bottom}{root}{a,b}
\EdgeL{root}{1}{a} 
\EdgeL{1}{8}{a} 
\EdgeL{8}{11}{b} 
\EdgeL{1}{2}{b} 
\EdgeR{2}{9}{a} 
\EdgeR{2}{4}{b} 
\EdgeR{root}{3}{b} 
\EdgeR{3}{7}{a} 
\EdgeR{7}{10}{a} 
\EdgeR{3}{5}{b} 
\EdgeR{5}{6}{a} 
%suffix pointers
\ChgEdgeLineStyle{dotted}
\LArcR[]{root}{bottom}{}
\ArcL[]{1}{root}{}
\ArcL[]{8}{1}{}
\ArcL[]{11}{2}{}
\ArcL[]{2}{3}{}
\ArcL[]{2}{3}{}
\ArcR[]{9}{7}{}
\ArcL[]{4}{5}{}
\ArcR[]{3}{root}{}
\ArcR[]{7}{1}{}
\ArcR[]{5}{3}{}
\VArcR[]{arcangle=-7}{10}{8}{}
\ArcR[]{6}{7}{}
\ChgEdgeLineStyle{solid}
\end{VCPicture}
}
\end{center}
\caption{Position heap for string $aababbbaabaab$ with suffix pointers
  (dotted arrows). Secondary positions are shown in italic.}
\label{PH-with-suf-point}
\end{figure}

We now begin to describe the on-line construction algorithm for
$\PH{T}$, given a text $T[1..n]$. Consider the node
$\overline{T[s..k]}$ of $\PH{T[1..k]}$ storing the active secondary position
$s$, that we call the {\em active node}. If the active secondary
position does not exist (i.e. there is no secondary positions at all),
then the active node is $\mathit{root}$ and the active position is set to
$k+1$. 
Observe that the nodes storing the other secondary positions
$s+1,s+2,\ldots,n$ can be reached, in order, by following the chain of
suffix pointers 
$f(\overline{T[s..n]}),f(f(\overline{T[s..n]})),\ldots$ 
until the root node is reached.  On the example of
Figure~\ref{PH-with-suf-point}, the active secondary position is $\mathit{12}$,
and the chain of suffix pointer outgoing from the active node leads to the node
storing position $\mathit{13}$ followed by the root. 

This brings us to the main trick of our construction: {\em we will not
  store secondary positions at all, but only 
memorize the active secondary position and the active node}. The
secondary positions can be easily recovered by traversing the chain of
suffix pointers starting from the active node and incrementing the
position counter after traversing each edge. Note also that if the input
string $T$ is ended by a unique sentinel symbol, the resulting
position heap does not contain any secondary nodes and there is no
need to recover them. 

%% Recall now that in $\PH{T[1..k]}$, all 
%% secondary positions $s,s+1,s+2,\ldots,k$ are stored, in order, in the nodes 
%% belonging to the chain of suffix links
%% $\overline{T[s..k]},f(\overline{T[s..k]}),f(f(\overline{T[s..k]})),\ldots$. 
%% Observe now ... Obviously, in this case
%% $f(\overline{T[i..k+1]})=\overline{T[i+1..k+1]}$. This shows that the
%% chain of suffix links can be split 
%% again into two intervals: nodes $\overline{T[s..k]},\ldots,
%% \overline{T[t-1..k]}$, $s<t\leq k$, without an outgoing $T[k+1]$-transition, and the
%% remaining nodes $\overline{T[t..k]},\ldots,\overline{T[k]},\mathit{root}$
%% having an outgoing $T[k+1]$-transition. 

Keeping in mind that the secondary positions are not stored explicitly,
the transformation of $\PH{T[1..k]}$ into $\PH{T[1..k+1]}$ specified
by Lemma~\ref{main} reduces to processing case \textit{(i)} only, as
case \textit{(ii)} does not imply any modification anymore. Case
\textit{(i)} is implemented by the
following simple procedure. Starting from the active node, the
algorithm traverses the 
chain of suffix pointers as long as the current node does not have an
outgoing $T[k+1]$-edge. For each such node, a new node is created
linked by a $T[k+1]$-edge to the current node. A suffix pointer to
this new node is set from the previously created new node. Once the
first node with an outgoing $T[k+1]$-edge is encountered, the
algorithm moves to the node this edge leads to, sets the suffix
pointer to this node, and assigns this node to be the active node for
the following iteration. The correctness of the last assignment is
stated in the following lemma.

\begin{lemma}
\label{new-active-node}
Consider $\PH{T[1..k]}$ and let $s$ be the active position, and $t\geq
s$ be the smallest position such that node $\overline{T[t..k]}$ has an
outgoing $T[k+1]$-edge. Then node $\overline{T[t..k+1]}$ is the active
node of $\PH{T[1..k+1]}$. 
\end{lemma}
\begin{proof}
As it follows from Lemma~\ref{main}, $t$ is the largest secondary
position of $T[1..k+1]$. 
\end{proof}

Algorithm~\ref{main-algo} provides a pseudo-code of the algorithm.
\begin{algorithm}
\begin{algorithmic}[1]
\STATE create states $\mathit{root}$ and $\perp$ 
\STATE $f(\mathit{root}) \gets \perp$
\FORALL {$a\in A$} \STATE set an $a$-edge from $\perp$ to $\mathit{root}$ \ENDFOR
\STATE $\mathit{currentnode} \gets \mathit{root}$
\STATE $\mathit{currentsuffix} \gets 1$
\FOR{$i = 1 \mathbf{~to~} n$}
   \STATE $\mathit{lastcreatednode} \gets \mathit{undefined}$
   \WHILE {$\mathit{currentnode}$ does not have an outgoing $T[i]$-edge}
      \STATE create a new node $\mathit{newnode}$ pointing to
      $\mathit{currentsuffix}$ \label{newnode1} 
      \STATE set a $T[i]$-edge from $\mathit{currentnode}$ to $\mathit{newnode}$
      \IF {$\mathit{lastcreatednode}\neq \mathit{undefined}$} \STATE $f(\mathit{lastcreatednode})\gets \mathit{newnode}$ \ENDIF
      \STATE $\mathit{lastcreatednode}\gets \mathit{newnode}$
      \STATE $\mathit{currentnode}\gets f(\mathit{currentnode})$
      \STATE $\mathit{currentsuffix}\gets \mathit{currentsuffix}+1$
   \ENDWHILE
   \STATE move $\mathit{currentnode}$ to the target node of the
   outgoing $T[i]$-edge 
   \IF {$\mathit{lastcreatednode}\neq \mathit{undefined}$} \STATE
   $f(\mathit{lastcreatednode})\gets \mathit{currentnode}$ \ENDIF 
\ENDFOR
\end{algorithmic}
\caption{On-line construction of the position heap $\PH{T[1..n]}$}\label{main-algo}
\end{algorithm}

The correctness of Algorithm~\ref{main-algo} follows from
Lemmas~\ref{main}, \ref{new-active-node} and the discussion above. It
is instructive, in addition, to observe the following: 
\begin{itemize}
\item it is easily seen that the suffix pointers of $\PH{T[1..k+1]}$ are
  correctly set. Indeed, the algorithm assigns to
  $\overline{T[i..k+1]}$ a suffix pointer to
  $\overline{T[i+1..k+1]}$ which is obviously correct. Note that for
  the active position $s$ of $T[1..k]$, the created node
  $\overline{T[s..k+1]}$ does not get pointed to by any suffix pointer,
  which is correct, as $T[s-1..k+1]$ is not represented in
  $\PH{T[1..k+1]}$: the position $s-1$ is primary in $T[1..k]$ and
  therefore the node $\overline{T[s-1..k]}$, if it exists in
  $\PH{T[1..k]}$, does not get extended by a $T[k+1]$-edge (cf
  Lemma~\ref{main}). 
\item since the depth of $\overline{T[s..k]}$ ($s$ is the active
  position) in $\PH{T[1..k]}$ is $k+1-s$ and a traversal of a suffix link
  decrements the depth by $1$ and increments the current position by
  $1$, it follows that if the traversal of the suffix chain reaches
  the root node, the active position value becomes $k+1$, which is
  exactly what we need to start processing the next letter
  $T[k+1]$. 
This shows why Algorithm~\ref{main-algo} correctly
 maintains $\mathit{currentsuffix}$ and never needs to reset it at
the beginning of the \textbf{for}-loop iteration. % (EXPLAIN??)
\end{itemize}

It is easy to see that the running time of Algorithm~\ref{main-algo}
is linear in the length $n$ of the input string. Since each iteration of
the \textbf{while}-loop creates a node, this loop iterates exactly $n$
times over the whole run of the algorithm. Trivially, the
\textbf{for}-loop iterates $n$ times too, and all the involved
operations are constant time. Thus, the whole algorithm takes $O(n)$
time. The following theorem concludes the construction.

\begin{theorem}
\label{conclusion}
For an input string $T[1..n]$, Algorithm~\ref{main-algo} correctly
constructs $\PH{T}$ on-line in time $O(n)$. 
\end{theorem}

\section{Augmented position heap}

Assume we have a text $T[1..n]$ for which we constructed the position
heap $\PH{T}$. We don't assume that $T$ is ended by a unique
letter, and therefore some nodes of $\PH{T}$ are double nodes and store 
two positions of $T$, one primary and one secondary. Here
we assume that the secondary positions {\em are} actually stored (or can be
retrieved in constant time for each node). As explained in Section~\ref{sect-algo},
even if the secondary positions are not stored during the construction
of $\PH{T}$, they can be easily recovered once the construction is
completed. 

\cite{Ehrenfeucht2011100} proposed a linear-time
string matching algorithm using $\PH{T[1..n]}$, i.e. an algorithm that
computes all occurrences of a pattern string in $T$ in time
$O(m+occ)$, where $m$ is the pattern length and $occ$ the number of
occurrences. Describing this elegant algorithm is beyond the scope of
this paper, we refer the reader to \cite{Ehrenfeucht2011100} for its
description. We only note that the algorithm itself applies without changes to
our definition of position heap, as it does not depend in any way on
the order that the suffixes of $T$ are inserted. 

However, the algorithm of \cite{Ehrenfeucht2011100} runs on $\PH{T}$
enriched with some additional information. Let $\overline{i}$ denote
the node of $\PH{T}$ storing position $i$, $1\leq i\leq n$. The
extended data structure, called the {\em augmented position heap},
should allow the following queries to be answered in constant time:
\begin{itemize}
\item given a position $i$, retrieve the node $\overline{i}$,
\item given two nodes $\overline{i}$ and $\overline{j}$, is 
$\overline{i}$ a (not necessarily immediate) ancestor of 
$\overline{j}$?
\item given a position $i$ of $T$, retrieve the node
  $\overline{T[i..i+\ell]}$, where $T[i..i+\ell]$ is the longest
  substring of $T$ starting at position $i$ and represented in
  $\PH{T}$. 
\end{itemize}

To answer the first query, \cite{Ehrenfeucht2011100} simply introduces
an auxiliary array storing, for each position $i$, a pointer to the node
$\overline{i}$. Maintaining this array during the construction of
$\PH{T}$ by Algorithm~\ref{main-algo} is trivial: once a position is
assigned to a newly created node (line~\ref{newnode1} 
of Algorithm~\ref{main-algo}), a new entry of the array is set. If
$T$ is not ended by a unique symbol and then the final $\PH{T}$ has
secondary positions, those are easily recovered by traversing the chain
of suffix pointers at the very end of the construction. 

The second query can be also easily answered in constant time after a
linear-time preprocessing of $\PH{T}$. A solution proposed in
\cite{Ehrenfeucht2011100} consists in traversing $\PH{T}$ depth-first
and storing, for each node, its discovery and finishing times
\cite{CLR}. Then node $\overline{i}$ is an ancestor of node
$\overline{j}$ if and only if the discovery and finishing time of
$\overline{i}$ is respectively smaller and greater than the discovery
and finishing time of $\overline{j}$. 

A more space-efficient solution would be to use a balanced parenthesis
representation of the tree topology of $\PH{T}$, taking $2n$ bits,
and link each 
node to the corresponding opening parenthesis. Then the corresponding
closing parenthesis can be retrieved in constant time by the method
of \cite{MunroRaman01} using $o(n)$ auxiliary bits. This allows
ancestor queries to be answered in constant time. 

The third type of queries is answered by an additional
mapping called {\em maximal-reach pointer}
\cite{Ehrenfeucht2011100}: for a position $i$ of $T[1..n]$, define
$mrp(i)$ to be the node $\overline{T[i..i+\ell]}$, where
$T[i..i+\ell]$ is the longest prefix of $T[i..n]$ 
% starting at position $i$ and 
represented in $\PH{T}$. Observe first that if $i$ is a secondary
position, then $mrp(i)=\overline{i}$. This is because 
a secondary position $i$ is stored in node $\overline{T[i..n]}$, which
trivially corresponds to the longest prefix starting at $i$. 
Therefore, as it is done in \cite{Ehrenfeucht2011100}, $mrp$ can be
represented by pointers from node $\overline{i}$ to node $mrp(i)$
whenever these nodes are different. In our case, we have then to keep
in mind that a maximal-reach pointer from a double node applies to the
primary position of this node. Figure~\ref{mrp-figure} provides an
illustration. 
\begin{figure}
\begin{center}
\VCDraw{%
\begin{VCPicture}{(0,-1)(16,8)}
% states
\MediumState
\State{(8,6)}{root} 
\State[1]{(3,4)}{1} 
\StateVar[2,\textit{12}]{(6,2)}{2}
\StateVar[3,\textit{13}]{(13,4)}{3}
\State[4]{(8,0)}{4}
\State[5]{(16,2)}{5}
\State[6]{(14,0)}{6}
\State[7]{(11,2)}{7}
\State[8]{(0,2)}{8}
\State[9]{(4,0)}{9}
\State[10]{(10,0)}{10}
\State[11]{(2,0)}{11}
%\State[12]{(6,0)}{12}
%\State[13]{(6,0)}{13}
% transitions
\EdgeL{root}{1}{a} 
\EdgeL{1}{8}{a} 
\EdgeL{8}{11}{b} 
\EdgeL{1}{2}{b} 
\EdgeR{2}{9}{a} 
\EdgeR{2}{4}{b} 
\EdgeR{root}{3}{b} 
\EdgeR{3}{7}{a} 
\EdgeR{7}{10}{a} 
\EdgeR{3}{5}{b} 
\EdgeR{5}{6}{a} 
%suffix pointers
\ChgEdgeLineStyle{dotted}
\ArcL[]{1}{root}{}
\ArcL[]{8}{1}{}
\ArcL[]{11}{2}{}
\ArcL[]{2}{3}{}
\ArcL[]{2}{3}{}
\ArcR[]{9}{7}{}
\ArcL[]{4}{5}{}
\ArcR[]{3}{root}{}
\ArcR[]{7}{1}{}
\ArcR[]{5}{3}{}
\VArcR[]{arcangle=-7}{10}{8}{}
\ArcR[]{6}{7}{}
\ChgEdgeLineStyle{solid}
\EdgeLineDouble
\FixEdgeLineDouble{1.0}{1.0}
%\SetEdgeLineDblStatus
\ArcR[]{1}{11}{}
\LArcR[]{8}{11}{}
\LArcL[]{2}{9}{}
\LArcL[]{3}{7}{}
\LArcL[]{7}{10}{}
%\RstEdgeLineDblStatus
\FixEdgeLineDouble{0.5}{0.6}
\EdgeLineSimple
\end{VCPicture}
}
\end{center}
\caption{Position heap for string $aababbbaabaab$ with suffix pointers
  and maximal-reach pointers $mrp$ (double arrows). 
% $mrp(i)$ is shown by an arrow from node $\overline{i}$. 
Only values for which $mrp(i)\neq \overline{i}$ are shown, namely
$mrp(1)=\overline{11}$, $mrp(8)=\overline{11}$, $mrp(2)=\overline{9}$,
$mrp(3)=\overline{7}$, $mrp(7)=\overline{10}$. Note that maximal reach pointers outgoing from double nodes
are unambiguous as for all secondary positions $i$, we have
$mrp(i)=\overline{i}$.} 
\label{mrp-figure}
\end{figure}

In \cite{Ehrenfeucht2011100}, maximal-reach pointers are computed by
an extra traversal of $\PH{T}$, using an auxiliary {\em dual heap}
structure on top of it (see Introduction). Here we show that
maximal-reach pointers can be easily computed using suffix pointers
instead of the dual heap. Thus, we completely get rid of the dual heap
for constructing the augmented position heap, replacing it with
suffix pointers. 

After $\PH{T}$ is constructed, we compute $mrp(i)$ iteratively for $i=1,2,\ldots,s-1$, where $s$ is
the active secondary position of $T[1..n]$. Assume we have computed
$mrp(i)$ for some $i$ and have to compute $mrp(i+1)$. Assume
$mrp(i)=\overline{T[i..i+\ell]}$. It is easily seen that
$T[i+1..i+\ell]$ is a prefix of the string represented by $mrp(i+1)$. To
compute $mrp(k+1)$, we follow the suffix link $f(mrp(k))$ to reach
$\overline{T[i+1..i+\ell]}$ and then keep extending the prefix
$T[i+1..i+\ell]$ as long as it is represented in $\PH{T}$. The
resulting pseudo-code is given in Algorithm~\ref{mrp-algo}. 
\begin{algorithm}
\begin{algorithmic}[1]
\STATE $\mathit{currentnode} \gets \mathit{root}$
\STATE $\mathit{readhead} \gets 1$
\FOR{$i = 1 \mathbf{~to~} n$}
   \WHILE {$\mathit{currentnode}$ has an outgoing
     $T[\mathit{readhead}]$-edge \AND $\mathit{readhead}\leq n$}
      \STATE move $\mathit{currentnode}$ to the target node of the
       outgoing $T[\mathit{readhead}]$-edge 
      \STATE $\mathit{readhead}\gets \mathit{readhead}+1$ \label{readhead-incr}
   \ENDWHILE
   \STATE $mrp(i)\gets \mathit{currentnode}$\label{mrp-set}
   \STATE $\mathit{currentnode} \gets f(\mathit{currentnode})$ \label{curnode-update}
\ENDFOR
\end{algorithmic}
\caption{Linear-time computation of maximal-reach pointers $mrp(i)$}\label{mrp-algo}
\end{algorithm}

It is easy to see that Algorithm~\ref{mrp-algo} works in time
$O(n)$: the $\mathbf{while}$-loop makes exactly $n$ iterations
overall, as each iteration increments the $\mathit{readhead}$
counter. 

The following property of Algorithm~\ref{mrp-algo} is useful to
observe: as soon as $\mathit{readhead}$ gets the
value $n+1$ (line~\ref{readhead-incr}), the node
$\mathit{currentnode}$ gets assigned to the active node of $\PH{T[1..n]}$
(line~\ref{curnode-update}); at the subsequent iterations, the
algorithm simply traverses the chain of suffix links and sets the
maximal-reach pointer for each secondary position to be the node
storing this position (lines~\ref{mrp-set}-\ref{curnode-update}). 

\newcommand{\rr}{\mathsf{rank}}
\newcommand{\sel}{\mathsf{select}}
Maximal-reach pointers constitute an additional data structure on top
of the tree structure of the position heap. However, it is interesting
to note that this structure can be represented compactly in $O(n)$
bits so that $mrp(i)$ can be computed in constant time. Here is how it
can be done. 

As observed earlier, $mrp(i)\geq mrp(i-1)-1$ for all
$i\in [2..n]$. Define $\delta_i=mrp(i)-mrp(i-1)+1$ and observe that $mrp(1)+\sum_{i=2}^{n}\delta_i=n$. Represent the
vector $(mrp(1),\delta_2,\ldots,\delta_{n})$ as a binary vector $\mathbb{B}_{mrp}$ by
representing all values in unary followed by a $0$. For example,
vector $(1,2,0,3,0,1,0)$ is then represented as $10110011100100$. Note that the
length of $\mathbb{B}_{mrp}$ is $2n$. To $\mathbb{B}_{mrp}$, we will be applying $\rr$ and $\sel$ operations. Recall that
for a binary vector $\mathbb{B}$, $\rr_1(\mathbb{B},i)$ returns the
number of $1$'s occurring in $\mathbb{B}[1..i]$, and
$\sel_1(\mathbb{B},\ell)$ returns the position of the $\ell$-th
occurrence of $1$ in $\mathbb{B}$ (counting from left). $\rr_0$ and
$\sel_0$ are defined similarly. It is known
that the
input binary vector of length $n$ can be pre-processed using $o(n)$ additional
memory bits, so that $\rr$
and $\sel$ queries can be answered in time $O(1)$
\cite{DBLP:conf/focs/Jacobson89,DBLP:conf/soda/ClarkM96}. 
Observe now that $mrp(i)=\rr_1(\sel_0(i))-i+1$. Therefore, $mrp(i)$
can be computed in constant time. We summarize this in the following
Lemma. 
\begin{lemma}
For the position heap $\PH{T}$ of any text $T[1..n]$, the
maximal-reach pointers $mrp(i)$, $i\in [1..n]$, can be 
stored in $2n+o(n)$ bits so that each $mrp(i)$ can be recovered in
time $O(1)$. 
\end{lemma}

\section{Concluding remarks}

We proposed a construction algorithm of a position heap of a
string, under a modified definition of position heap compared to
\cite{Ehrenfeucht2011100}. In contrast with the algorithm of
\cite{Ehrenfeucht2011100} that processes the sequence right-to-left, our
algorithm reads the string left-to-right and has the on-line
property. Drawing a parallel to suffix trees, our algorithm can be compared to Ukkonen's
on-line algorithm \cite{Ukkonen93}, while the
algorithm of \cite{Ehrenfeucht2011100} can be compared to Weiner's
algorithm \cite{Weiner73}. The similarity of our algorithm to
Ukkonen's algorithm goes beyond this parallel, as the execution of 
both algorithms (e.g. the way of traversing the tree under
construction, or updating the active node)
are clearly analogous. 
% are also somewhat analogous in their design. 

Position heap is a smaller data structure than suffix tree: it
contains exactly $n+1$ nodes whereas the suffix tree has $n$ leaves and
then up to $2n$ nodes. Still, the position heap allows for a
linear-time string matching. The position heap is a new data
structure and many questions about it are open. 

The $O(n)$ complexity bounds of both Algorithm~\ref{main-algo}
(Theorem \ref{conclusion}) and Algorithm~\ref{mrp-algo} are stated
for a constant-size alphabet, otherwise a correcting factor $\log |A|$
should be introduced, similarly to the suffix tree construction. 
One may ask if there exists a linear-time 
% (i.e. time-independent of the alphabet size) 
algorithm (not necessarily on-line) which constructs position heaps within
a time independent of the alphabet size, as Farach's algorithm does
for suffix trees \cite{DBLP:conf/focs/Farach97}. 

An interesting direction to study is whether the position heap can
be compacted. The theory of compact data structures has become a major
subfield of string processing, and has accumulated a number of
interesting and powerful techniques
\cite{DBLP:journals/csur/NavarroM07}. In this paper, we showed that
some components of the position heap can be effectively compacted,
however the compaction of its main part (the trie) is still to be
studied. 

It would be interesting to study combinatorial properties of position
heaps. In combinatorial terminology, a tree with $n$ nodes labeled by
distinct integers from $\{1,\ldots,n\}$ and such that the label of any
node is smaller than the label of any of its descendant is called an
{\em increasing tree}. It is known, for example, that there are $n!$
ordered binary increasing trees \cite{stanley1999enumerative}
(``ordered'' means that left and right children are distinguished),
which implies that there are $(n+1)!$ ``potential position heaps''
over binary alphabet. Obviously, there are only $2^n$ different
position heaps over the binary alphabet. It would be interesting to
establish combinatorial properties that distinguish arbitrary
increasing trees from position heaps. 

The authors of \cite{Ehrenfeucht2011100} proposed algorithms
for updating the position heap when the input string undergoes
modifications (character insertions/deletions). We believe that these
algorithms can be easily applied to our definition of position
heap. Recently, the authors of \cite{NakashimaEtAlSPIRE12} showed that
the position heap can be generalized to a set of strings stored
in a trie such that the construction and pattern matching remain linear-time. 
Other interesting applications of position heap are still to
be discovered. 

%% It would be interesting to study further the properties
%% of maximal-reach pointers. Note that their structure differs between 
%% our definition of position heap and the definition of
%% \cite{Ehrenfeucht2011100}. 
From a more practical perspective, it would be also interesting to exploit the
``adaptiveness'' of position heaps to substring frequencies, mentioned
in Section~\ref{sect-properties}. 

\bibliography{../biblio,../pat_match}

\end{document}